\documentclass[sigconf, nonacm]{acmart}

\usepackage{algorithm}
\usepackage[noend]{algpseudocode}
\usepackage{amsmath}
\usepackage{enumitem}
\usepackage{booktabs}
\usepackage{subcaption}
\usepackage{graphicx}
\usepackage[capitalize,nameinlink]{cleveref}

\newcommand\vldbdoi{XX.XX/XXX.XX}
\newcommand\vldbpages{XXX-XXX}
\newcommand\vldbvolume{14}
\newcommand\vldbissue{1}
\newcommand\vldbyear{2020}
\newcommand\vldbauthors{\authors}
\newcommand\vldbtitle{\shorttitle} 
\newcommand\vldbavailabilityurl{https://github.com/WildAlg/threshold-algorithm}
\newcommand\vldbpagestyle{plain}

\newcommand{\revfix}[1]{{#1}} 
\begin{document}
\title{Scalable Triangle Counting: The Threshold Algorithm}

\newcommand{\E}{\mathbb{E}}
\newcommand{\Prb}{\mathbb{P}}
\newcommand{\Var}{\mathrm{Var}}

\author{Asaf Etgar}
\affiliation{%
  \institution{Yale University}
  \city{New Haven}
  \country{USA}
}
\email{asaf.etgar@yale.edu}

\author{Anna Gilbert}
\orcid{0000-0002-1825-0097}
\affiliation{%
  \institution{Yale University}
  \city{New Haven}
  \country{USA}
}
\email{anna.gilbert@yale.edu}

\author{Quanquan C. Liu}
\orcid{???}
\affiliation{%
  \institution{Yale University}
  \city{New Haven}
  \country{USA}
}
\email{quanquan.liu@yale.edu}

\author{Andrew McGregor}
\orcid{???}
\affiliation{%
  \institution{University of Massachusetts}
  \city{Amherst}
  \country{USA}
}
\email{amcgrego@umass.edu}

\begin{abstract}
We study one-pass triangle counting on random-order edge streams. We
present a remarkably simple algorithm---read edges from the stream
until $Q$ triangles are observed in the prefix, then output
$Q\,(m/S)^3$ where $S$ is the stopping length---and prove that, when
the maximum number of triangles incident to any edge satisfies
$\eta \le T^{2/3}$, this is a $(1\pm\varepsilon)$-approximation of
$T$ with probability $1-\delta$ using
$O(\varepsilon^{-2}\log(1/\delta)\, m/T^{1/3})$ memory. Crucially, the
algorithm does not need any a priori estimate of $T$, in sharp
contrast with state-of-the-art sampling-rate based algorithms
\cite{mcgregor2020triangle, TsourakakisKangMillerFaloutsos2009}. It
also does not need a prescribed memory budget: the stopping rule
self-selects the prefix length and can return an estimate before reading the
entire stream.

The proof rests on a Schudy--Sviridenko concentration argument for an
independent-edge-sampling estimator, coupled to the without-replacement
prefix produced by the algorithm. On six real temporal streams, the
algorithm's stopping prefix follows the predicted cube-root scaling and
achieves at most $6\%$ error at a $10\%$ prefix, without using $T$. At a
fixed stored-edge budget, variance-reduced reservoir samplers are often
more accurate, but only after reading the entire stream. On a separate, much larger, 
$1.8\times10^9$-edge graph, the threshold algorithm reads $0.46\%$ of
the stream and returns $3.8\%$ error, while the strongest reservoir
baselines do not finish a pass within the wall-clock cap.
\end{abstract}

\maketitle

\pagestyle{\vldbpagestyle}
\begingroup\small\noindent\raggedright\textbf{PVLDB Reference Format:}\\
\vldbauthors. \vldbtitle. PVLDB, \vldbvolume(\vldbissue): \vldbpages, \vldbyear.\\
\href{https://doi.org/\vldbdoi}{doi:\vldbdoi}
\endgroup
\begingroup
\renewcommand\thefootnote{}\footnote{\noindent
This work is licensed under the Creative Commons BY-NC-ND 4.0 International License. Visit \url{https://creativecommons.org/licenses/by-nc-nd/4.0/} to view a copy of this license. For any use beyond those covered by this license, obtain permission by emailing \href{mailto:info@vldb.org}{info@vldb.org}. Copyright is held by the owner/author(s). Publication rights licensed to the VLDB Endowment. \\
\raggedright Proceedings of the VLDB Endowment, Vol. \vldbvolume, No. \vldbissue\ %
ISSN 2150-8097. \\
\href{https://doi.org/\vldbdoi}{doi:\vldbdoi} \\
}\addtocounter{footnote}{-1}\endgroup

\ifdefempty{\vldbavailabilityurl}{}{
\vspace{.3cm}
\begingroup\small\noindent\raggedright\textbf{PVLDB Artifact Availability:}\\
The source code, data, and/or other artifacts have been made available at \url{\vldbavailabilityurl}.
\endgroup
}

\section{Introduction}\label{sec:intro}

Counting the number of triangles $T$ in a massive graph $G=(V,E)$ is
one of the most-studied primitives in scalable graph data
science~\cite{baryossef2002streamingtriangles,buriol2006datastreamtriangles,braverman2013streaminghard,jha2013birthday,pavan2013vldb,cormode2014secondlook,McGregorVorotnikovaVu2016,destefani2016triest,jayaram2021optimal}. In the single-pass streaming model, an algorithm
sees the edges of $G$ one at a time and must output an estimate of
$T$ using memory much smaller than $|E|=m$. Sampling primitives are the standard design paradigm in this setting.
There is also a long line of work on counting triangles via sublinear time algorithms that, potentially in addition to sampling, may make various queries to the graph \cite{eden2015sublineartriangles,assadi2019subgraphcounting,bera2020tetris,bishnu2025arboricity,tetek2022fmm}.

Worst-case lower bounds
make this problem challenging: even on random-order streams, any
$(1+\varepsilon)$-approximation requires
$\Omega(\varepsilon^{-2} m / \sqrt{T})$ memory when $T \le
\sqrt{m}$~\cite{mcgregor2020triangle}. State-of-the-art random-order
algorithms~\cite{mcgregor2020triangle,
TsourakakisKangMillerFaloutsos2009,
TsourakakisKolountzakisMiller2011} attain this bound up to log
factors, and all utilize a sample-based approach: They set a sampling rate $p$ (or multiple rates) and sample edges as they arrive. Crucially, they are all \emph{$T$-aware} and require an \emph{a priori} constant-factor estimate of
$T$ to set their sampling rate. This requirement is prohibitive on
real-world streams: even two snapshots of the same network can differ
in $T$ by an order of magnitude. 

Another family of triangle counting algorithms are \emph{$T$-free} and do not require an estimate on $T$~\cite{destefani2016triest,LimKang2015mascot,Jung2016furl,PaghTsourakakis2012colorful,ShinWRS2017,ShinThinkD2018,Wu2025great}. These algorithms are provided with a specified memory budget $M$, and utilize reservoir or budget sampling to keep a variance-reducing set, typically alongside a running unbiased triangle count. While these families of algorithms do not assume a priori knowledge of $T$, both families suffer from \emph{scalability} and \emph{adaptability} issues. First, they require a full pass over the entire stream of edges. When graphs grow in orders of magnitude, even a single pass could be the runtime bottleneck for an algorithm, even if all other parameters are set optimally. Secondly, they rely on input parameters to guarantee accuracy; e.g., a sufficiently large memory budget. 


This paper studies a remarkably simple alternative: An algorithm that returns a global estimate of $T$ while only looking at a small prefix of the stream and without an initial estimate of $T$ or a memory budget. In a random order stream, any prefix consists of a uniform sample of edges drawn without replacement. Consequently, rather than forcing a full pass over the stream, our algorithm halts when a prefix is large enough to serve as a signal for the entire graph.  A major benefit of the algorithm is that said prefix is \emph{self-selected} by the algorithm and does not need to be prescribed, while still guaranteeing a tight approximation of $T$. The algorithm is surprisingly simple to implement, and different implementations allow the user to adjust the space-time-accuracy tradeoff according to use case constraints.

More formally, let $G=(V,E)$ be a
simple graph with $m=|E|$ edges and $T$ triangles. For each edge $e$,
let $\tau(e)$ denote the number of triangles containing $e$, and let
\[
\eta := \max_{e\in E} \tau(e).
\]
Throughout, we assume $0<\varepsilon\le 1/2$ and $0<\delta<1/10$.
We consider the stochastic process of drawing random edges from $G$
without replacement until some prescribed number $Q$ of triangles is
observed (assuming $Q\leq T$). Let $S$ be the random variable
corresponding to the number of edges sampled by this process.

\begin{theorem} \label{thm:main}
Assume $\eta \leq T^{2/3}$ and set $Q=c_0 \varepsilon^{-6} \log^3 (4/\delta)$ for a sufficiently large constant $c_0>0$. Then
\[ \widehat T:=Q \left(\frac{m}{S }\right)^3 \]
is a $(1\pm \varepsilon)$ approximation of $T$ with probability at least $1-\delta$. Furthermore, $S=O(\varepsilon^{-2} \log (1/\delta)\, m/T^{1/3})$ with  probability at least $1-\delta$.
\end{theorem}

The argument extends to arbitrary $\eta$, but then the choice of
$Q$ depends\footnote{One can resolve this limitation by using a learning-augmented predictor to estimate this ratio from initial stream data. This method is studied in the companion paper~\cite{liu2025learning}, which uses prefix bucket profiles and
deep neural networks to estimate triangle and $4$-cycle counts on
graph families in both random and arbitrary orders.} on the unknown ratio $\eta/T$. 

Our algorithm uses $\widetilde{O}(m/T^{1/3})$ space. This is a factor $\approx T^{1/6}$ more than the space used by McGregor and Vorotnikova \cite{mcgregor2020triangle} (henceforth MV20) near both the optimal $m/\sqrt{T}$ and the  lower bound $\Omega(\varepsilon^{-2}m/\sqrt{T})$ of \cite{mcgregor2020triangle} (for $T\le \sqrt{m}$). 
The experiments show that this premium buys a different operating point:
the algorithm self-selects $1$--\revfix{$6.6\%$} prefixes with single-digit error
on real temporal streams, reads only $0.46\%$ of the $1.8\times10^9$-edge
\emph{com-friendster} stream with $3.8\%$ error, and converts the main
work into parallel static triangle counting on the stored prefix. Thus
the tradeoff is qualitative rather than only asymptotic: instead of
spending a full pass and a stale estimate of $T$, the threshold rule
allocates space from the observed stream and returns before the input is
exhausted.

\paragraph{Contributions.}
\begin{itemize}[topsep=2pt,itemsep=0pt,leftmargin=*]
\item \cref{sec:algorithm} presents the streaming algorithm
that realizes \cref{thm:main}
(\cref{alg:threshold}). The algorithm has only a single
tunable parameter, $Q$, and does \emph{not} require any a priori
estimate of $T$ or space budget.
\item \cref{sec:concentration}--\cref{sec:stopping} prove
\cref{thm:main} via Schudy--Sviridenko concentration on the
independent-edge-sampling estimator, coupled to the
without-replacement prefix produced by the algorithm. As a byproduct of our analysis, we improve the bound on the sampling probability $p$ by polylogarithmic factors, thereby improving the analysis for DOULION \cite{TsourakakisKangMillerFaloutsos2009} and the subsequent work on triangle sparsifiers \cite{TsourakakisKolountzakisMiller2011}.
\item \cref{sec:experiments} evaluates the algorithm on six
real temporal streams and two large-scale graphs (by order of magnitude). The threshold $Q$
acts as a smooth self-adapting knob: accuracy increases with $Q$ and our Threshold algorithm 
reads only $0.46\%$--\revfix{$6.6\%$} of the stream,
with single-digit error. Against MV20~\cite{mcgregor2020triangle}
and a broad set of $T$-free baselines, Threshold is within $0.08$
absolute error of the best equal-space method. The
main advantage of our algorithm is the early-stop feature:
forced to stop at the Threshold prefix,
native reservoir samplers undershoot the triangle count
with error $\approx1$ and need
$97$--$98\%$ of the stream to reach $5$--$10\%$ error, while Threshold
returns after reading a very small prefix, 
reading up to $215\times$ fewer edges on
\emph{com-friendster}, and its GBBS~\cite{dhulipala2018theoretically}
implementation scales to a
$25\times$ speedup on $64$ cores.
\end{itemize}



\section{The Threshold Algorithm}\label{sec:algorithm}

We now describe the algorithm that realizes the stochastic process
analyzed in \cref{thm:main}. The algorithm receives a
random-order stream of $m$ edges from a graph $G$ and is
given a single integer parameter $Q$, the \emph{triangle threshold}.
It keeps every arriving edge in memory and incrementally counts
triangles closed by each new edge; once $Q$ triangles have been
observed it stops storing edges and returns its estimate.

\begin{algorithm}[t]
\caption{Threshold algorithm for random-order triangle counting.}
\label{alg:threshold}
\begin{algorithmic}[1]
\Require Stream length $m$; triangle threshold $Q$; random-order stream
of edges of $G$.
\State $\mathcal{S} \gets \emptyset$;\ adjacency list $A \gets \emptyset$;\ $\widehat t \gets 0$.
\For{each arriving edge $e=(u,v)$}
  \State $\widehat t \gets \widehat t + |N_A(u)\cap N_A(v)|$
        \Comment{triangles closed by inserting $e$}
  \State $\mathcal{S} \gets \mathcal{S}\cup\{e\}$;\ update $A$ to include $e$.
  \If{$\widehat t \ge Q$} \textbf{break} \EndIf
\EndFor
\State $S \gets |\mathcal{S}|$.
\State \Return $\widehat T = Q \cdot \left(\dfrac{m}{S}\right)^3$.
\end{algorithmic}
\end{algorithm}

The implementation maintains $\mathcal{S}$ as a hash set and $A$ as
hash-based adjacency lists; each edge insertion costs
$O\big(\deg_A(u)+\deg_A(v)\big)$ time, so the total running time is
dominated by the cost of the triangle enumeration inside the prefix of
length $S$. Memory usage is $O(S)$ machine words for the stored edges
plus $O(S)$ for the adjacency structure.

\subsection{Theoretical guarantees}\label{sec:alg-guarantees}

\cref{alg:threshold} is the algorithm to which
\cref{thm:main} directly refers. We restate the guarantee in
terms of the algorithm's parameters.
\begin{theorem}[Restatement of \cref{thm:main}]
\label{thm:alg-main}
Assume $\eta \le T^{2/3}$ and set $Q = c_0\, \varepsilon^{-6}\log^3(4/\delta)$ for a sufficiently large constant $c_0>0$. Let $S$ be the stopping time of \cref{alg:threshold}. Then $\widehat T = Q\,(m/S)^3$ satisfies
\[
\Pr\bigl[\, (1-\varepsilon)\,T \;\le\; \widehat T \;\le\; (1+\varepsilon)\,T \,\bigr]\;\ge\; 1-\delta,
\]
and $S = O\bigl(\varepsilon^{-2}\log(1/\delta)\cdot m / T^{1/3}\bigr)$
with probability at least $1-\delta$.
\end{theorem}
\cref{sec:concentration} and \cref{sec:stopping} give the proof. In \cref{sec:concentration}, we analyze sufficient conditions for tight concentration of the sampled number of triangles in the independent-edge sampling setup via the Schudy–Sviridenko inequality \cite{SS}. In \cref{sec:stopping}, we analyze the algorithm's stopping time $S$. Because the prefix of a random order stream behaves like a uniform sample without replacement of the edges, bounding $S$ directly is challenging. To overcome this, we couple the prefix with two binomial random variables $X_1$ and $X_2$ which represent prefix lengths drawn at an over-sampling rate $r_1$ and an under sampling rate $r_2$. Crucially, a random-length prefix behaves identically to analyzing independent edge sampling allowing us to utilize results from~\cref{sec:concentration}. By carefully selecting the rates and the threshold $Q$, we are able to couple $S$ with $X_1,X_2$, and consequently guarantee both tight concentration of $S$ and of the estimator $\widehat{T}$. Finally, we prove \cref{thm:alg-main} in the parameter-free form under the promise $\eta \le T^{2/3}$.

\paragraph{Practical setting of $Q$.}
The constant $c_0$ in \cref{thm:alg-main} is non-explicit, so
rather than fixing $Q$ as a number we fix the \emph{space budget}
directly and let it induce $Q$. Given a target memory fraction $f$, we
run the loop of \cref{alg:threshold} until exactly
$S=\lceil f m\rceil$ edges have been stored, then report
$\widehat T = \widehat t\,(m/S)^3$, where $\widehat t$ is the number of
triangles closed within the stored prefix. This is exactly
\cref{alg:threshold} run with the \emph{implicit} threshold
$Q=\widehat t$: because a uniform random prefix of fraction $f$ contains
$\approx f^3 T$ triangles, the realized threshold concentrates at
$Q\approx f^3 T$. Sweeping the memory budget $f$ is therefore
equivalent to sweeping $Q$, and---crucially---requires no knowledge of
$T$, since $f$ is a quantity the operator already controls. (Deployed
the other way around, with a fixed $Q$ and no budget, the stopping
time $S$ instead \emph{self-adjusts} to the data and attains the
$O(\varepsilon^{-2}\log(1/\delta)\,m/T^{1/3})$ bound of
\cref{thm:alg-main} automatically.) In the implementation the
target $S=\lceil f m\rceil$ is reached by a geometric doubling
schedule; on the dense real streams the process halts essentially at
$\lceil f m\rceil$, while on sparse streams the realized $S$ can exceed
the target before $Q$ triangles accumulate, so we always report the
\emph{realized} $S/m$.

\paragraph{Heavy-edge regime.}
When $\eta > T^{2/3}$, \cref{thm:main} no longer applies in the
parameter-free form, but the analysis of
\cref{sec:concentration}--\cref{sec:stopping} still gives a
$(1\pm\varepsilon)$-approximation provided $Q$ is set to absorb the
factor $\eta/T$. The ratio $\eta/T$ can be estimated from a small
prefix using learned features of the stream; we focus the experiments
below on the parameter-free regime and leave a full learning-augmented
evaluation to companion work.

\section{Concentration of Triangles via Independent Edge Sampling}\label{sec:concentration}

Fix a sampling rate $p\in(0,1]$, and sample each edge independently with probability $p$.
For each edge $e\in E$, let $X_e\in\{0,1\}$ be the indicator that $e$ is sampled.
Define the triangle counting polynomial
\[
Y \;=\; \sum_{\Delta} \prod_{e\in \Delta} X_e,
\]
where the sum is over all triangles $\Delta$ in $G$.
Then $Y$ is the number of sampled triangles, and
\[
\hat T := \frac{Y}{p^3}
\]
is the natural unbiased estimator for $T$. The goal of this section is to derive a sufficient condition on $p$ ensuring $
\Prb\left[\, |\hat T-T| \ge \varepsilon T \,\right] \le \delta$.

\subsection{The Schudy--Sviridenko inequality}

We use the following result by Schudy and Sviridenko  \cite[Theorem 1.4]{SS}.

\begin{theorem}[Schudy--Sviridenko]\label{thm:SS}
Let $f(Y_1,\dots,Y_n)$ be a polynomial of degree $q$ and maximal variable power $\Gamma$ in independent moment-bounded random variables $Y_1,\dots,Y_n$, all with the same moment-boundedness parameter $L$.
Let $\mu_r$ denote the maximum expected partial derivative of order $r$ and  $\mu_0=\E[f(Y)]$.
Then there is an absolute constant $R\ge 1$ such that, for every $\lambda>0$,

\begin{align*}
&\Prb\left[\, |f(Y)-\E f(Y)| \ge \lambda \,\right] \le\\
&
e^2 \cdot
\max\left(
\max_{1\le r\le q} \exp\left\{-\frac{\lambda^2}{\mu_0\mu_r L^r \Gamma^r R^q}\right\},
\;
\max_{1\le r\le q} \exp\left\{-\left(\frac{\lambda}{\mu_r L^r \Gamma^r R^q}\right)^{1/r}\right\}
\right) \ .
\end{align*}
\end{theorem}

\subsection{Applying the concentration result to triangle counting}
We wish to bound the parameters $\mu_r$ of \cref{thm:SS} for the triangle counting polynomial $Y$. Since no edge can appear twice in a single triangle and each triangle is comprised of $3$ edges, $Y$ is multilinear of degree $3$. That is, $\Gamma = 1$ and $q = 3$. The random variables $X_e$ are Bernoulli, $|X_e|\le 1$ and therefore a uniform moment bound is  $L = 1$. For $1\le r \le q = 3$ we compute a bound on $\mu_r$.


\begin{lemma}
For the triangle-counting polynomial
\[
Y=\sum_{\Delta}\prod_{e\in \Delta} X_e,
\]
the Schudy--Sviridenko parameters satisfy $
\mu_0 = p^3 T,
\mu_1 = p^2 \eta,
\mu_2 \le p,$ and $\mu_3 \le 1$.
\end{lemma}

\begin{proof}
First, 
\[
\mu_0 = \E[Y] = \sum_{\Delta} \E\left[\prod_{e\in \Delta} X_e\right] = p^3 T.
\]
For $r=1$, differentiating with respect to an edge-variable $X_e$ gives
\[
\frac{\partial Y}{\partial X_e}
=
\sum_{\Delta\ni e}\prod_{f\in \Delta\setminus\{e\}} X_f.
\]
Here, the notation $\Delta\ni e$ means counting over all triangles containing the edge $e$. Taking expectations,
\[
\E\left[\frac{\partial Y}{\partial X_e}\right]
=
\sum_{\Delta\ni e} p^2
=
p^2 \tau(e).
\]
Therefore $
\mu_1 = \max_e p^2 \tau(e) = p^2 \eta$. 

For $r=2$, fix distinct edges $e,h$.
Then
$
\frac{\partial^2 Y}{\partial X_e\,\partial X_h}
$
is nonzero if and only if $e$ and $h$ share a common triangle.
Since the graph is simple, two edges share at most one triangle.
Hence
\[
\E\left[\frac{\partial^2 Y}{\partial X_e\,\partial X_h}\right]
\le p,
\]
for all pairs of edges $e,h$, and thus $\mu_2\le p$.

For $r=3$, recall that $Y$ is multilinear of degree $3$ with coefficient $1$ for all summands. Therefore the third mixed derivative is $1$ if and only if the three differentiated edges form a triangle and $0$ otherwise.
Hence $\mu_3\le 1$.
\end{proof}

We proceed to apply \cref{thm:SS} according to the bounds computed.
Set $
\lambda := \varepsilon \mu_0 = \varepsilon p^3 T$
and then \[
\Prb\left[\, |\hat T-T| \ge \varepsilon T \,\right]
=
\Prb\left[\, |Y-\E Y| \ge \varepsilon \E Y \,\right]
=
\Prb\left[\, |Y-\mu_0| \ge \lambda \,\right].
\]

 We set
\[
A_r := \frac{\lambda^2}{\mu_0 \mu_r R^3},
\qquad
B_r := \left(\frac{\lambda}{\mu_r R^3}\right)^{1/r}
\]
where $R$ is the constant from \cref{thm:SS}.
In this notation, \cref{thm:SS} states
\[
    \Prb\left[\, |Y-\mu_0| \ge \lambda \,\right]
\le
e^2 \exp\left(
- \min\{A_1,A_2,A_3,B_1,B_2,B_3\}
\right),    
\]
Using the upper bounds on $\mu_r$, we obtain the following lower bounds on $A_r,B_r$:
\begin{align*}
A_1 &\ge \frac{\varepsilon^2 pT}{\eta R^3}, & A_2 &\ge \frac{\varepsilon^2 p^2 T}{R^3}, & A_3 &\ge \frac{\varepsilon^2 p^3 T}{R^3},\\
B_1 &\ge \frac{\varepsilon pT}{\eta R^3},   & B_2 &\ge \left(\frac{\varepsilon p^2 T}{R^3}\right)^{1/2}, & B_3 &\ge \left(\frac{\varepsilon p^3 T}{R^3}\right)^{1/3}.
\end{align*}
%

Since $\varepsilon \le 1$, $B_1 \ge A_1$. Since $p\le 1$ we have $A_2\ge A_3$. Moreover, 
\[\left(\varepsilon p^2T/R^3\right)^{1/2}
\ge
\left(\varepsilon p^3T/R^3\right)^{1/2}
\ge
\left(\varepsilon p^3T/R^3\right)^{1/3}\] and therefore $B_2\ge B_3$.
Thus the inequality simplifies to
\begin{eqnarray}
\label{eq:main-tail}
& & \Prb\left[\, |\hat T-T| \ge \varepsilon T \,\right]\nonumber \\
& \le &
e^2\exp\left(-\min \left\{A_1,A_3,B_3\right\}\right)\nonumber\\
& = &
e^2 \exp\left(
- \min\left\{
\frac{\varepsilon^2 pT}{\eta R^3},
\frac{\varepsilon^2 p^3 T}{R^3},
\left(\frac{\varepsilon p^3 T}{R^3}\right)^{1/3} 
\right\}
\right).
\end{eqnarray}
We arrive at the following sufficient condition for concentration of $\hat T$ around $T$.

\begin{theorem}\label{thm:indp}
Let $G$ be a simple graph with $T$ triangles and 
\[
\eta = \max_{e\in E}\tau(e).
\]
Sample each edge independently with probability $p$, let $Y$ be the number of sampled triangles, and define $\hat T=Y/p^3$. There is an absolute constant $C>0$ such that if $p\geq p^\star$ where 
\begin{equation} \label{eq:pstar}
p^\star (\varepsilon,\delta) \;=\;
\max\left\{
\frac{C\,\eta\,L_\delta}{\varepsilon^2 T},
\;
\left(\frac{C\,L_\delta}{\varepsilon^2 T}\right)^{1/3},
\;
\left(\frac{C\,L_\delta^3}{\varepsilon T}\right)^{1/3}
\right\}
\end{equation}
where $L_\delta=2+\log(1/\delta)$. Then $ \Prb\left[\, |\hat T-T| \ge \varepsilon T \,\right] \le \delta$.
\end{theorem}

\begin{proof}
A sufficient condition for \cref{eq:main-tail} to be bounded above by $\delta$ is:
\begin{align}
\min\left (\frac{\varepsilon^2 pT}{\eta R^3} , 
\frac{\varepsilon^2 p^3 T}{R^3} , 
\left(\frac{\varepsilon p^3 T}{R^3}\right)^{1/3} \right )\ge L_\delta. \label{eq:cond3}
\end{align} 
Take $C = R^3$ where $R$ is the constant from~\autoref{thm:SS}. Rearranging the three lower bounds on $p^\star$ in \eqref{eq:pstar} gives exactly the displayed minimum condition in \eqref{eq:cond3}. 
\end{proof}

\section{Stopping Times for Edge Sampling without Replacement}\label{sec:stopping}






The goal of this section is to apply the independent sampling result from \cref{sec:concentration} to a sampling without replacement setting. Fix a uniformly random ordering of the \(m\) edges of \(G\).
For \(t\in\{0,1,\dots,m\}\), let \(E(t)\) denote the set of the first \(t\) edges in this ordering, and let
$
N(t)
$
denote the number of triangles contained in \(E(t)\).
We wish to  analyze the concentration of the stopping time
\[
S :=\min\{\, t : N(t)\ge Q \,\}.
\]
for an appropriate choice of a threshold $Q$. Bounding $S$ directly is difficult as edges in $E(t)$ are not drawn independently. In order to utilize \Cref{thm:indp}, we look for a sampling rate $r$ with the following properties:
\begin{enumerate}
    \item The rate $r$ is large enough to guarantee tight concentration around $T$ as in \Cref{thm:indp}.
    \item The threshold $Q$ is a reliable indicator that we have processed an $r$-fraction of the stream, and not much more. 
\end{enumerate}
Then, instead of directly analyzing $S$, we couple $S$ with two random variables $X_1,X_2$ with two distinct sampling rates $r_1,r_2$. We think of $r_1$ as an over sampling rate, and $r_2$ as an under sampling rate. Crucially, taking a prefix of a random length $X_i$ behaves like independent edge sampling from $E$, allowing us to utilize \autoref{thm:indp}. Our choice of a threshold $Q$ will separate these two sampling rates, and imply strong concentration of $S$. We proceed to formalize this approach.

\subsection{Coupling the Stopping Time}

Let \(p^\star\) denote the critical value of \(p\) from \cref{thm:indp}.  Let 
\[r\geq (1-\varepsilon)^{-1}  \cdot \max(p^\star,
3\varepsilon^{-2} \log (1/\delta)/m) \ .\]
Define two sampling rates
\[r_1:= (1+\varepsilon)\, r
\quad \mbox{ and } \quad r_2 := (1-\varepsilon)\, r . \]


Consider two random variables \(X_1\sim \mathrm{Bin}(m,r_1)\) and \(X_2\sim \mathrm{Bin}(m,r_2)\). 

\begin{lemma}\label{lem:binom to ind}
For each \(i\in\{1,2\}\), the random set \(E(X_i)\) has the same distribution as an independent \(r_i\)-sample of the edge set \(E\).
\end{lemma}

\begin{proof}
Fix \(i\in\{1,2\}\) and a subset \(F\subseteq E\) of size \(k\).
Then
\begin{eqnarray*}
\Prb[E(X_i)=F]
& = & 
\Prb[X_i=k]\cdot \frac{1}{\binom{m}{k}} \\
& = &
\binom{m}{k}r_i^k(1-r_i)^{m-k}\cdot \frac{1}{\binom{m}{k}} \\
& =& 
r_i^k(1-r_i)^{m-k},
\end{eqnarray*}
which is exactly the probability that \(F\) is obtained by keeping each edge independently with probability \(r_i\).
\end{proof}
Consequently, the number of triangles in \(E(X_i)\) has the same distribution as the random variable \(Y\) from \cref{sec:concentration} with sampling rate \(r_i\). Set \(Q=r^3 T\). We show that the choice of $Q$ separates the two sampling rates with probability at least \(1-\delta\). Denote by $\Prb_{p = r_i}$ the probability space where each edge in $E$ is sampled independently with probability $r_i$.

\begin{lemma}\label{lem:Qdivide} 
$
\Prb_{p=r_1}[Y> Q] \ge 1-\delta$
and 
$\Prb_{p=r_2}[Y<Q] \ge 1-\delta$.
\end{lemma}

\begin{proof}
Recall that \cref{thm:indp} applies at any sample rate that is at least $p^\star$, and that $r_1,r_2 \ge (1-\varepsilon)r \ge p^\star$. Then for $r_1$
by \cref{thm:indp} with probability at least \(1-\delta\),
\[
Y\ge (1-\varepsilon)r_1^3 T = (1-\varepsilon)(1+\varepsilon)^3 Q  >  Q,
\]
which proves the first inequality. Similarly, for $r_2$ with probability at least \(1-\delta\),
\[
Y\le (1+\varepsilon)r_2^3 T  = (1+\varepsilon)(1-\varepsilon)^3 Q < Q
\ .
\]
proving the second inequality.
\end{proof}

Combined with Lemma \ref{lem:binom to ind}, Lemma \ref{lem:Qdivide} transfers immediately to random prefixes.

\begin{lemma}\label{lem:NX}
$
\Prb\bigl[N(X_1)> Q\bigr]\ge 1-\delta$
 and 
$\Prb\bigl[N(X_2)<Q\bigr]\ge 1-\delta
$. 
\end{lemma}
\begin{proof}
    By Lemma \ref{lem:binom to ind},
    \[
    \Prb\bigl[N(X_1)> Q \bigr] = \Prb_{p=r_1}\bigl[Y > Q \bigr]\ge 1-\delta,
    \]
    and a similar argument holds for $N(X_2)$.
\end{proof}

%
%

We can now bound the stopping time $S$.

\begin{corollary}[Concentration of $S$] \label{cor:almostthere}
\[
\Pr[mr(1-\varepsilon)^2 < S < mr(1+\varepsilon)^2]\geq 1-4\delta \ .
\]
\end{corollary}

\begin{proof}
Recall that $S = \min\{t : N(t) \ge Q\}$. By Lemma \ref{lem:NX}, $N(X_1) > Q$ with probability at least $1-\delta$, and therefore $X_1 \ge S$. Similarly, $X_2 \le S$. with probability at least $1-\delta$ By the union bound
\[\Prb\left[X_2 < S  < X_1\right]\ge 1-2\delta \ .\] 
Applying the standard multiplicative Chernoff bound to \(X_1 \sim \mathrm{Bin}(m,r_1)\) gives
\[
\Prb\left[X_1 < (1+\varepsilon)r_1m\right] \ge 1-\delta 
\]
and similarly
\[
\Prb\left[X_2 > (1-\varepsilon)r_2m\right] \ge 1-\delta.
\]
Taking the union bound over both gives
\[
\Prb\left[\, 
(1-\varepsilon ) r_2m \leq X_2 \mbox{ and } X_1\leq (1+\varepsilon ) r_1m 
\right]\ge 1-2\delta.
\]
Plugging in the definition of $r_1,r_2$ gives
\[
\Prb\left[\, 
(1-\varepsilon )^2 rm \leq X_2 \mbox{ and } X_1\leq (1+\varepsilon )^2 rm 
\right]\ge 1-2\delta.
\]
Alongside $\Prb\left[X_2 < S  < X_1\right]\ge 1-2\delta$, the union bound implies the corollary.
\end{proof}
\subsection{Designing a Triangle Estimator}
By Corollary \ref{cor:almostthere}, the number of edges one must process before seeing \(Q\) triangles is tightly concentrated around $mr$. We wish to use $S$ as an estimator for the underlying sampling rate $r$. Since $
T={Q}/{r^3}$
and  \(S \approx mr\), this motivates the following  estimate of the number of triangles in \(G\): 

\[
\widehat T:=Q\left(\frac{m}{S }\right)^3.
\]

\begin{corollary}[Triangle Estimation]\label{cor:triangleestimation}
\[
\Pr \left [
\frac{T}{(1+\varepsilon)^6}
\;\le\;
\widehat T
\;\le\;
\frac{T}{(1-\varepsilon)^6} \right ]\geq 1-4\delta \ .
\]
\end{corollary}

\begin{proof}
Since
\[
\widehat T
=
Q\left(\frac{m}{S }\right)^3
=
r^3T\left(\frac{m}{S }\right)^3
=
\left(\frac{mr}{S }\right)^3 T,
\]
By  Corollary \ref{cor:almostthere}, with probability at least $1-4\delta$
\[
mr(1-\varepsilon)^2 < S \le mr(1+\varepsilon)^2.
\]
So with probability at least $1-4\delta$
\[
\frac{T}{(1+\varepsilon)^6} \le \left(\frac{mr}{S }\right)^3T \le  \frac{T}{(1-\varepsilon)^6}. \qedhere
\] 
\end{proof}

\subsection{Proof of Main Theorem}

\begin{proof}[Proof of \cref{thm:main}]

Recall that 
\[
p^\star (\varepsilon,\delta) \;=\;
\max\left\{
\frac{C\,\eta\,L_\delta}{\varepsilon^2 T},
\;
\left(\frac{C\,L_\delta}{\varepsilon^2 T}\right)^{1/3},
\;
\left(\frac{C\,L_\delta^3}{\varepsilon T}\right)^{1/3}
\right\}
\]
If $\eta \leq T^{2/3}$, all three terms are $O(\varepsilon^{-2} \log (1/\delta)/T^{1/3})$ and so  \[p^\star=O(\varepsilon^{-2} \log (1/\delta)/T^{1/3}) \ . \]
%
Hence, we set $r= c_0^{1/3}\varepsilon^{-2} \log (1/\delta)/T^{1/3}$ for sufficiently large $c_0$. Since $T^{1/3}\le m$,
\begin{flalign*}
& \frac{c_0^{1/3}\varepsilon^{-2}\log(1/\delta)}{T^{1/3}} \ge \frac{c_0^{1/3}\varepsilon^{-2}\log(1/\delta)}{m}
\end{flalign*}
and therefore 
\[r\geq (1-\varepsilon)^{-1} \max(p^\star, 3 \varepsilon^{-2} \log (1/\delta)/m) \ .\]
Consequently
\[
r^3T=\left(\frac{c_0^{1/3}\varepsilon^{-2} \log (1/\delta)}{T^{1/3}}\right)^3T = c_0\varepsilon^{-6}\log^3(1/\delta).
\]
Setting $Q = r^3T = c_0\varepsilon^{-6}\log^3(1/\delta) $ and $\widehat{T} = \left(\frac{m}{S}\right)^3Q$, by Corollary \ref{cor:triangleestimation}
\[
\Pr \left [
\frac{T}{(1+\varepsilon)^6}
\;\le\;
\widehat T
\;\le\;
\frac{T}{(1-\varepsilon)^6} \right ]\geq 1-4\delta \ .
\]
Reparameterizing by $\varepsilon \leftarrow \varepsilon/10$ and $\delta\leftarrow \delta/4$ gives the result.
\end{proof}

\section{Experimental Evaluation}\label{sec:experiments}

The experiments test the main promise of the paper: the threshold algorithm
chooses a short prefix on its own and estimates $T$ accurately by \emph{only} looking at the
prefix of the stream.
With a fixed stored-edge budget, the best reservoir
samplers are often more accurate because they produce an estimate
after reading the \emph{entire} stream; 
the threshold algorithm is the \emph{first early-stopping} algorithm that achieves
comparable approximations looking at \emph{only} a subset of the stream.
Our advantage is that the algorithm
stops early, needs no estimate of $T$, and still returns a whole-graph
estimate.

\paragraph{Experimental Setup.}
We use six temporal streams from SNAP~\cite{leskovec2014snap}, the
Network Repository~\cite{networkrepo}, and KONECT~\cite{konect}
(\emph{copresence-InVS15}, \emph{reddit},
\emph{sx-superuser}, \emph{wiki-talk-temporal}, \emph{ca-cit-HepPh}
(the KONECT co-citation network; $m=3{,}148{,}447$), and
\emph{sx-stackoverflow}), plus \emph{com-orkut} and
\emph{com-friendster} for large-scale analysis. 
We remove self-loops and duplicate
edges. All accuracy experiments use uniformly random edge orders,
matching \cref{thm:alg-main}; the true triangle count is used
only for scoring. \cref{alg:threshold} is implemented in
GBBS~\cite{dhulipala2018theoretically} with
ParlayLib~\cite{parlaylib}; ground-truth counts use
Shun--Tangwongsan~\cite{ST15}. We re-implemented the remaining baselines:
TRI\`EST~\cite{destefani2016triest}, ThinkD~\cite{ShinThinkD2018},
WRS~\cite{ShinWRS2017}, GREAT~\cite{Wu2025great}, and
MASCOT~\cite{LimKang2015mascot}; DOULION~\cite{TsourakakisKangMillerFaloutsos2009},
MV20~\cite{mcgregor2020triangle}, and triangle
sparsification~\cite{TsourakakisKolountzakisMiller2011}; wedge
and colorful samplers~\cite{seshadhri2013wedge,PaghTsourakakis2012colorful};
and related estimators~\cite{buriol2006datastreamtriangles,jha2013birthday,pavan2013vldb,Jung2016furl}.
All use the same C++/GBBS setup so every method sees identical input,
representation, timing, and error accounting. Unless stated otherwise,
accuracy entries report $\mathbb{E}_{\pi}[|\widehat T_{\pi}-T|/T]$,
where $\pi$ ranges over independent random stream orders (every trial
draws a fresh seeded permutation; no permutation is reused across
operating points or across algorithms). We use: $30$ random orders for
fixed-prefix accuracy rows (including the err.@$10\%$ column); $10$ per
$Q$ for self-sizing; $10$ per operating point for MV20; $20$ orders for
matched-space and iso-accuracy rows; and $5$ for stream-read.
Experiments run on shared-cluster nodes with two $32$-core Intel Xeon
Gold 8562Y+ CPUs, one thread per core, and $1$--$4$ TB of DDR5.
The real streams and \emph{com-orkut} use a $1$ TB node;
\emph{com-friendster} uses $4$ TB. We compile the GBBS/ParlayLib code
with \texttt{g++}. 

\begin{table}[!tb]\centering
\caption{Graphs used in the evaluation. The first six are real temporal
streams; the last two are static graphs streamed in uniformly random
order for scalability.}
\label{table:graph-info}
\small
\begin{tabular}{|l|r|r|r|}
\toprule
Graph & vertices & edges & triangles \\
\midrule
copresence-InVS15 & 219 & 16{,}725 & 713{,}002\\
reddit & 35{,}776 & 124{,}330 & 406{,}391\\
sx-superuser & 192{,}409 & 714{,}570 & 1{,}543{,}161\\
wiki-talk-temporal & 1{,}094{,}018 & 2{,}787{,}967 & 8{,}113{,}676\\
ca-cit-HepPh & 28{,}093 & 3{,}148{,}447 & 195{,}758{,}685\\
sx-stackoverflow & 2{,}584{,}164 & 28{,}183{,}518 & 114{,}206{,}974\\
\midrule
com-orkut & 3{,}072{,}441 & 117{,}185{,}083 & 627{,}584{,}181\\
com-friendster & 65{,}608{,}366 & 1{,}806{,}067{,}135 & 4{,}173{,}724{,}142\\
\bottomrule
\end{tabular}
\end{table}

\subsection{Self-Sizing and Accuracy}
\label{sec:exp-adaptive}

\cref{table:adaptive} gives selected operating points while \cref{fig:qsweep}
shows the full sweep. The result is that $Q$ behaves as a \emph{smooth accuracy
knob}: increasing $Q$ lowers error while increasing the self-selected
prefix, and $S/m$ matches predicted space usage (labeled ``pred.''). 
The algorithm takes only the triangle threshold $Q$: it reads until 
the prefix contains $Q$ triangles, stops at $S$ edges, and returns $Q(m/S)^3$. Since
a random $f$-prefix contains about $f^3T$ triangles, we expect
$S/m \approx (Q/T)^{1/3}$. We
choose the smallest swept $Q$ with single-digit mean error, reading only
$0.9$--$9.2\%$ of each stream. The last column is a separate $10\%$-prefix
run ($30$ random orders); for \emph{ca-cit-HepPh} and
\emph{sx-stackoverflow}, this is much larger than the threshold
$1.0\%$ and $0.9\%$ prefixes that were used by our algorithm.

\begin{table}[!htp]\centering
\caption{Self-sizing on real streams. $Q$ is the target number of prefix
triangles; all other entries are percentages. The final column is a
separate fixed-$10\%$-prefix check.}
\label{table:adaptive}
\small
\setlength{\tabcolsep}{3pt}
\begin{tabular}{|l|r|r|r|r|r|}
\toprule
Stream & $Q$ & pred. & $S/m$ & err. & err.@$10\%$ \\
\midrule
copresence & $50$ & $4.1$ & $6.2$ & $7.8$ & $2.6$ \\
reddit & $50$ & $5.0$ & $6.6$ & $4.7$ & $6.1$ \\
sx-superuser & $500$ & $6.9$ & $9.2$ & $3.8$ & $3.2$ \\
wiki-talk & $100$ & $2.3$ & $3.3$ & $4.8$ & $1.5$ \\
ca-cit-HepPh & $50$ & $0.6$ & $1.0$ & $5.5$ & $0.3$ \\
sx-stackoverflow & $50$ & $0.8$ & $0.9$ & $7.2$ & $0.7$ \\
\bottomrule
\end{tabular}
\end{table}

\begin{figure}[htbp]\centering
\includegraphics[scale=0.65]{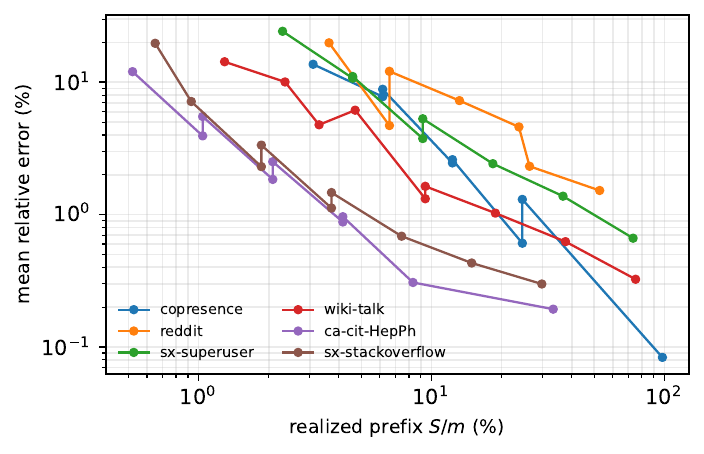}
\caption{Sweep of $Q$ values. Error decreases while the
fraction of space usage increases as $Q$ grows. The algorithm
is given $Q$ only, never $T$ or a memory budget.}
\label{fig:qsweep}
\end{figure}

\subsection{Comparison with State-of-the-Art}
\label{sec:exp-matched}

\cref{fig:triangle-stream} compares against
MV20~\cite{mcgregor2020triangle}, the state-of-the-art random-order
streaming algorithm with the best asymptotic bounds.
MV20 is given the true value of $T$; our algorithm is
not. Asymptotically MV20 uses nearly the optimal amount of space;
empirically, however, its realized footprint is governed by a
$T$-dependent floor that is independent of the accuracy parameter
$\varepsilon$: across our streams the (corrected) implementation stores
$49$--$77\%$ of all edges regardless of the requested budget ($65\%$ on
\emph{com-orkut}, $77\%$ on \emph{sx-stackoverflow}). 
Given that space,
MV20 with the true $T$ is accurate--often more accurate than the
threshold algorithm at a matched nominal $\varepsilon$. However, it stores
roughly an order of magnitude more edges than our self-selected prefix at
comparable error. Moreover, a single MV20 run on the $1.8$-billion-edge
\emph{com-friendster} did not complete within $10$ hours at $16$ cores,
whereas the threshold algorithm finishes every operating point in
seconds (\cref{sec:exp-large}).
The threshold algorithm reaches the
single-digit accuracy regime without a triangle-count estimate while
reading and storing only a small prefix.

\begin{figure}[!htb]\centering
    \includegraphics[width=\columnwidth]{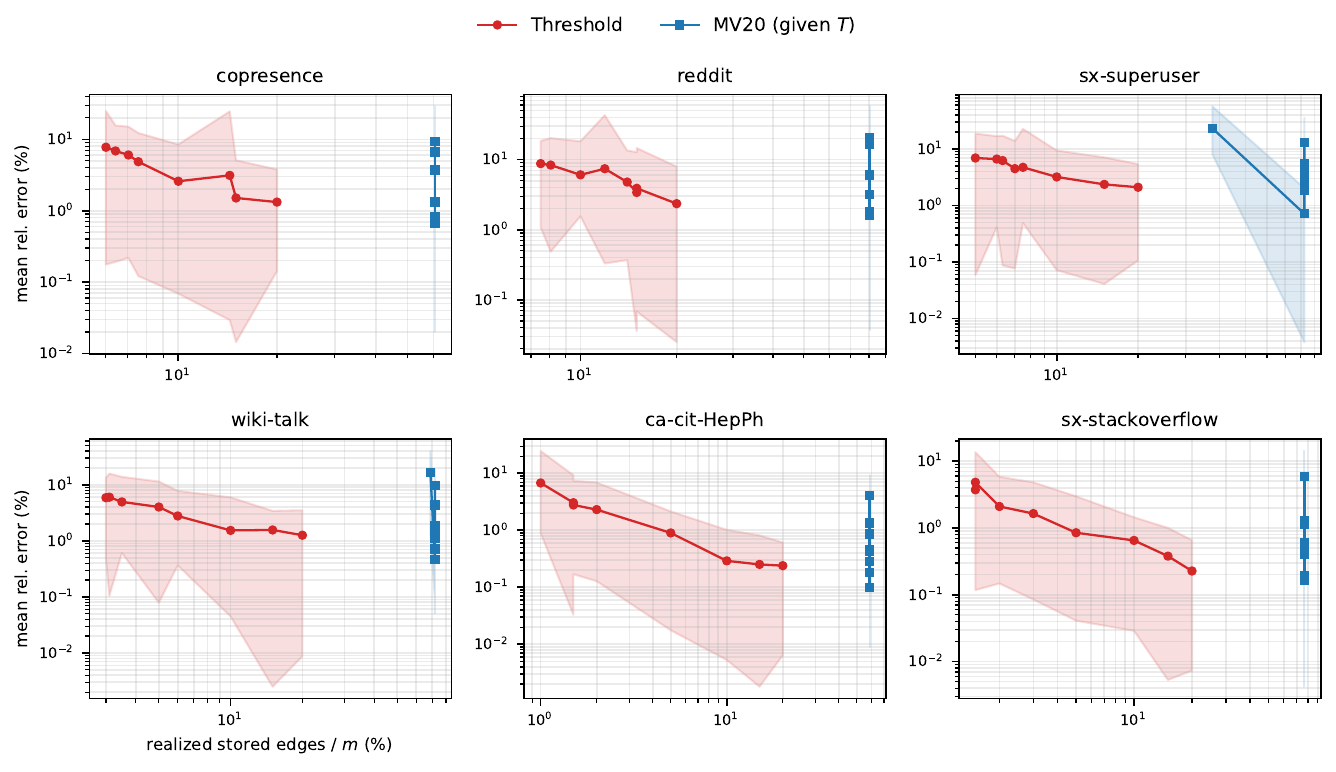}
    \caption{Accuracy vs.\ space usage for the threshold algorithm
    (red, swept over $Q$) and MV20 (blue, swept over $\varepsilon$ and
    given $T$).}
    \label{fig:triangle-stream}
\end{figure}

\cref{table:matched} merges the two fixed-space views. In each cell, the
left value is the mean relative error when the baseline receives exactly the
threshold-selected budget $M=S$ and no $T$; the right value is the smallest
budget fraction that matches the best Threshold accuracy target. Thus the
matched-space comparison is exact ($1.0\times$ the same stored-edge budget),
and Threshold is within $0.08$ absolute error of the best entry in every
column. The same-accuracy side shows the complementary tradeoff: the
variance-reduced reservoir samplers can match the target with less space on
several streams, and Threshold uses at most $12\times$ the smallest matching
space, but those methods still spend that memory after reading all $m$ edges.
Threshold spends more variance, and sometimes more prefix space, to avoid the
whole-stream pass.

\paragraph{Books stress test.}
We also keep the controlled heavy-edge experiment under the neutral name
\emph{books}: each book has one spine edge shared by many triangular
pages, and disjoint books are mixed with isolated triangles to sweep
$\rho=\eta^3/T^2$ from $10^{-7}$ to $10$. At a $5\%$ stored-edge budget,
the experiment cleanly explains the theorem's heavy-edge condition. When
$\rho\!\ll\!1$, triangles are spread out and wedge sampling is best
($0.007$ error versus Threshold's $0.031$ at $\rho=10^{-5}$). Near the
boundary the picture reverses: Threshold beats the best space-only
baseline at $\rho=0.1$ ($0.131$ versus $0.245$) and at $\rho=1$
($0.238$ versus $0.366$). Past the boundary all methods degrade. The
point is not that the condition is cosmetic; it is that the transition is
visible, and the threshold rule is competitive precisely at that
transition.
\begin{table*}
[!htb]\centering
\caption{Matched-space and same-accuracy comparison. Each cell reports
\emph{error at the threshold-selected budget $M=S$} / \emph{minimum space
needed to match the Threshold target accuracy}. The Threshold row gives its
matched-space error and the target error@space used for the second quantity.
The columns are the streams common to both sweeps. Reservoir
methods may win either half, but only after reading the whole stream; Threshold
returns from the prefix. Bold marks the best value within each half of a
column.}\label{table:matched}
\footnotesize
\setlength{\tabcolsep}{1.6pt}
\begin{tabular*}{\textwidth}{@{\extracolsep{\fill}}|l|c|c|c|c|c|}
\toprule
Algorithm & sx-superuser & wiki-talk & cit-HepPh & reddit & copresence \\
\midrule
stored edges $S$ / target & 32,781 / $0.017@18\%$ & 131,087 / \revfix{$0.018@9\%$} & 32,781 / $0.010@4\%$ & 8,203 / $0.048@13\%$ & 1,032 / $0.024@12\%$ \\
\midrule
\emph{Threshold (ours)} & $0.065\,/\,18\%$ & \revfix{$0.039\,/\,9\%$} & $0.056\,/\,4\%$ & $0.094\,/\,13\%$ & $0.095\,/\,12\%$ \\
\midrule
BuriolPODS~\cite{buriol2006datastreamtriangles} & $2.261\,/\,{>}20\%$ & \revfix{$2.824\,/\,{>}20\%$} & $0.079\,/\,{>}20\%$ & $0.868\,/\,{>}20\%$ & $0.051\,/\,{>}20\%$ \\
ColorfulTC~\cite{PaghTsourakakis2012colorful} & $0.051\,/\,20\%$ & \revfix{$0.023\,/\,10\%$} & $0.018\,/\,5\%$ & $0.046\,/\,15\%$ & $0.058\,/\,{>}20\%$ \\
Doulion~\cite{TsourakakisKangMillerFaloutsos2009} & $0.073\,/\,{>}20\%$ & \revfix{$0.028\,/\,15\%$} & $0.074\,/\,5\%$ & $0.074\,/\,20\%$ & $0.059\,/\,{>}20\%$ \\
FURL~\cite{Jung2016furl} & $0.087\,/\,20\%$ & \revfix{$0.035\,/\,10\%$} & $0.048\,/\,10\%$ & $0.071\,/\,15\%$ & $0.107\,/\,{>}20\%$ \\
GreatI~\cite{Wu2025great} & $0.030\,/\,20\%$ & \revfix{$\mathbf{0.010}\,/\,\mathbf{5\%}$} & $0.011\,/\,2\%$ & $0.040\,/\,10\%$ & $0.033\,/\,20\%$ \\
GreatII~\cite{Wu2025great} & $\mathbf{0.024}\,/\,\mathbf{10\%}$ & \revfix{$0.011\,/\,\mathbf{5\%}$} & $0.010\,/\,\mathbf{1\%}$ & $\mathbf{0.024}\,/\,10\%$ & $0.036\,/\,10\%$ \\
GreatPlus~\cite{Wu2025great} & $0.025\,/\,15\%$ & \revfix{$\mathbf{0.010}\,/\,\mathbf{5\%}$} & $0.011\,/\,2\%$ & $0.043\,/\,10\%$ & $0.056\,/\,20\%$ \\
MV20~\cite{mcgregor2020triangle} & $0.073\,/\,{>}20\%$ & \revfix{$0.028\,/\,15\%$} & $0.074\,/\,5\%$ & $0.074\,/\,20\%$ & $0.059\,/\,{>}20\%$ \\
Martingale~\cite{destefani2016triest} & $0.028\,/\,15\%$ & \revfix{$0.017\,/\,\mathbf{5\%}$} & $0.006\,/\,\mathbf{1\%}$ & $0.044\,/\,10\%$ & $\mathbf{0.015}\,/\,5\%$ \\
MascotC~\cite{LimKang2015mascot} & $0.073\,/\,{>}20\%$ & \revfix{$0.028\,/\,15\%$} & $0.074\,/\,5\%$ & $0.074\,/\,20\%$ & $0.059\,/\,{>}20\%$ \\
MascotI~\cite{LimKang2015mascot} & $0.038\,/\,{>}20\%$ & \revfix{$0.023\,/\,10\%$} & $0.016\,/\,5\%$ & $0.048\,/\,10\%$ & $0.033\,/\,15\%$ \\
NaiveScale (control) & $0.108\,/\,{>}20\%$ & \revfix{$0.032\,/\,10\%$} & $0.062\,/\,5\%$ & $0.113\,/\,15\%$ & $0.072\,/\,15\%$ \\
ThinkDAcc~\cite{ShinThinkD2018} & $0.028\,/\,15\%$ & \revfix{$0.017\,/\,\mathbf{5\%}$} & $0.006\,/\,\mathbf{1\%}$ & $0.044\,/\,10\%$ & $\mathbf{0.015}\,/\,5\%$ \\
ThinkDFast~\cite{ShinThinkD2018} & $0.028\,/\,15\%$ & \revfix{$0.017\,/\,\mathbf{5\%}$} & $\mathbf{0.005}\,/\,\mathbf{1\%}$ & $0.044\,/\,10\%$ & $\mathbf{0.015}\,/\,5\%$ \\
TriangleSparsifier~\cite{TsourakakisKolountzakisMiller2011} & $0.073\,/\,{>}20\%$ & \revfix{$0.028\,/\,15\%$} & $0.074\,/\,5\%$ & $0.074\,/\,20\%$ & $0.059\,/\,{>}20\%$ \\
TriestBase~\cite{destefani2016triest} & $0.091\,/\,{>}20\%$ & \revfix{$0.052\,/\,15\%$} & $0.068\,/\,5\%$ & $0.083\,/\,20\%$ & $0.051\,/\,15\%$ \\
TriestImpr~\cite{destefani2016triest} & $0.028\,/\,15\%$ & \revfix{$0.017\,/\,\mathbf{5\%}$} & $0.006\,/\,\mathbf{1\%}$ & $0.044\,/\,10\%$ & $\mathbf{0.015}\,/\,5\%$ \\
WRS~\cite{ShinWRS2017} & $0.030\,/\,15\%$ & \revfix{$0.018\,/\,10\%$} & $0.009\,/\,\mathbf{1\%}$ & $0.045\,/\,\mathbf{5\%}$ & $\mathbf{0.015}\,/\,5\%$ \\
WedgeSampling~\cite{seshadhri2013wedge} & $0.053\,/\,15\%$ & \revfix{$0.066\,/\,{>}20\%$} & $0.007\,/\,\mathbf{1\%}$ & $0.031\,/\,\mathbf{5\%}$ & $\mathbf{0.015}\,/\,\mathbf{1\%}$ \\
\bottomrule
\end{tabular*}
\end{table*}

\subsection{Early Stopping: The Previously Missing Domain}
\label{sec:exp-early}

We now investigate the effect of our key novelty on accuracy and space usage.
\cref{table:earlystop} shows that the threshold algorithm
returns an estimate after seeing only $0.46\%$--{$6.6\%$} of the stream.
Reservoir samplers, budget samplers, and related streaming estimators~\cite{destefani2016triest,LimKang2015mascot,Jung2016furl,ShinWRS2017,ShinThinkD2018,Wu2025great,TsourakakisKangMillerFaloutsos2009,McGregorVorotnikovaVu2016,TsourakakisKolountzakisMiller2011,seshadhri2013wedge,PaghTsourakakis2012colorful,buriol2006datastreamtriangles,jha2013birthday,pavan2013vldb}
are also single-pass algorithms, but they must see all $m$ edges before returning the estimate. This is the resource that matters when the input itself is the
bottleneck: on \emph{com-friendster}, the threshold algorithm reads
$215\times$ fewer edges before returning an estimate.

\begin{table}[!ht]\centering
\caption{Early stopping. The threshold algorithm halts after $S$ edges; any reservoir/sampling baseline must see all $m$ edges. On com-friendster this is a $215\times$ stream-read reduction.}\label{table:earlystop}
\small
\setlength{\tabcolsep}{2.5pt}
\begin{tabular}{|l|r|r|r|r|}
\toprule
Stream & $m$ & reads $S$ & $S/m$ & $m/S$ \\
\midrule
copresence & 16,725 & 1,032 & 6.2\% & $16\times$ \\
reddit-hyperlink & 124,330 & 8,203 & 6.6\% & $15\times$ \\
sx-superuser & 714,570 & 32,781 & 4.6\% & $22\times$ \\
wiki-talk & \revfix{2,787,967} & 131,087 & \revfix{4.7\%} & \revfix{$21\times$} \\
cit-HepPh & 3,148,447 & 32,781 & 1.0\% & $96\times$ \\
sx-stackoverflow & 28,183,518 & 367,018 & 1.3\% & $77\times$ \\
com-orkut & 117,185,083 & 1,048,594 & 0.9\% & $112\times$ \\
com-friendster & 1,806,067,135 & 8,388,629 & 0.5\% & $215\times$ \\
\midrule
\emph{any reservoir/sampling method} & $m$ & $m$ & $100\%$ & $1\times$ \\
\bottomrule
\end{tabular}

\end{table}

We next deny the reservoir samplers their whole-stream pass. In
\cref{table:stopearly}, every method is forced to stop at the same
$S$-edge prefix selected by the threshold rule and is given no $T$. The
off-the-shelf estimators report only the graph they have seen, so their
whole-graph estimate undershoots to error $\approx1$. The $+\,(m/S)^3$
column is deliberately oracle-aided: a random $f=S/m$ prefix has about
$f^3T$ triangles, so it scales the prefix count by $(m/S)^3$. Knowing that
the selected prefix is meaningful would require the triangle-count scale,
essentially the precise value of $T$, so this gives the reservoir samplers
far more power than their native interface. Useful early stopping is not
just storing a prefix; it is knowing how to extrapolate it without a
triangle-count estimate.

\begin{table}[!htb]\centering
\caption{Forced early stop at our self-selected prefix $S$, with budget $M=S$ and no $T$. Native samplers report the prefix graph and undershoot to error $\approx1$; adding our cubic extrapolation (which requires precise knowledge of $T$) makes them comparable. The threshold algorithm provides the whole-graph estimate. Best of extrapolated sampler and Threshold is \textbf{bold}.}
\label{table:stopearly}
\small
\setlength{\tabcolsep}{2pt}
\begin{tabular}{|l|r|r|r|r|}
\toprule
Stream & $S/m$ & samplers @ $S$ & $+\,(m/S)^3$ & \emph{Threshold (ours)} \\
\midrule
copresence & $6.2\%$ & $1.00$ & $\mathbf{0.083}$ & $0.095$ \\
reddit & $6.6\%$ & $1.00$ & $\mathbf{0.087}$ & $0.094$ \\
sx-superuser & $4.6\%$ & $1.00$ & $0.075$ & $\mathbf{0.065}$ \\
wiki-talk & \revfix{$4.7\%$} & $1.00$ & \revfix{$0.040$} & \revfix{$\mathbf{0.039}$} \\
cit-HepPh & $1.0\%$ & $1.00$ & $0.068$ & $\mathbf{0.056}$ \\
stackoverflow & $1.3\%$ & $1.00$ & $\mathbf{0.037}$ & $0.075$ \\
com-orkut & $0.9\%$ & $1.00$ & $0.055$ & $\mathbf{0.024}$ \\
com-friendster & $0.46\%$ & $1.00$ & $0.043$ & $\mathbf{0.038}$ \\
\bottomrule
\end{tabular}
\end{table}

\begin{figure}[!htb]\centering
\includegraphics[width=\columnwidth]{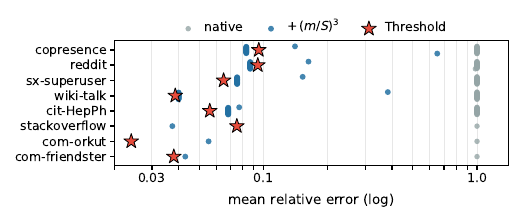}
\caption{Early stopping at the threshold prefix $S$. Reservoir
samplers report the prefix graph and undershoot to error $\approx1$;
cubic extrapolation makes them comparable to the threshold
algorithm, although the extrapolation still underperforms much of the time.
The threshold algorithm is the method that provides this
whole-graph estimate natively, with no $T$ and no prescribed budget.}
\label{fig:earlystop}
\end{figure}

Finally, \cref{table:streamread} removes memory as the bottleneck
and asks how much of the stream must be read before a whole-graph
estimate is within $10\%$ or $5\%$ error. A native reservoir sampler,
even with unlimited memory, must read $97$--$98\%$ of the stream because
a random $f$-prefix contains only about $f^3T$ triangles. The threshold
algorithm reaches comparable accuracy after a $1$--{$7\%$} prefix because
it stops when it finds enough triangles and consequently extrapolates that number in its prefix to the entire graph.

\begin{table}[!htb]\centering
\caption{Stream read before a whole-graph estimate is within $\varepsilon$, with unlimited memory. Native reservoir samplers must read almost all edges because an $f$-prefix contains about $f^3T$ triangles; cubic extrapolation helps only after being handed the stopping fraction.}
\label{table:streamread}
\small
\setlength{\tabcolsep}{2pt}
\begin{tabular}{|l|c|cc|cc|}
\toprule
 & Threshold & \multicolumn{2}{c|}{Reservoir (native)} & \multicolumn{2}{c|}{Reservoir $+(m/S)^3$} \\
Stream & reads (err) & $\le\!10\%$ & $\le\!5\%$ & $\le\!10\%$ & $\le\!5\%$ \\
\midrule
sx-superuser & 5\%~($0.065$) & 97\% & 98\% & 4\% & 6\% \\
wiki-talk & \revfix{5\%~($0.039$)} & 97\% & 98\% & \revfix{2\%} & \revfix{4\%} \\
cit-HepPh & 1\%~($0.056$) & 97\% & 98\% & 1\% & 1\% \\
reddit & 7\%~($0.094$) & 97\% & 98\% & 5\% & 11\% \\
copresence & 6\%~($0.095$) & 97\% & 98\% & 2\% & 7\% \\
\bottomrule
\end{tabular}
\end{table}

\subsection{Runtime Scalability}
\label{sec:exp-large}
\begin{figure}[!hb]\centering
\vspace{-0.5\baselineskip}
\includegraphics[scale=.6]{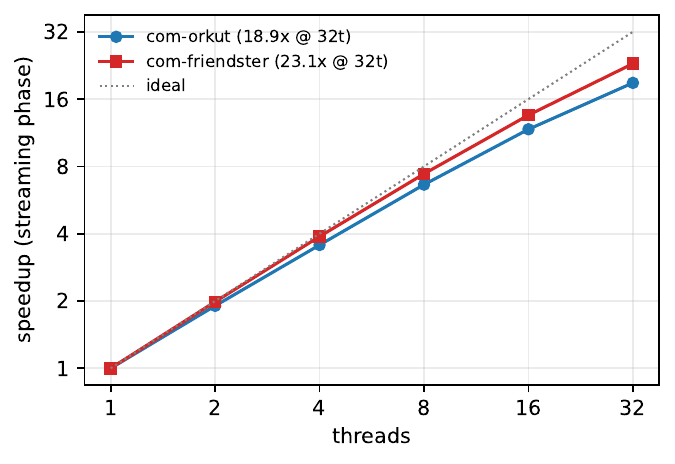}
\Description{Line plot of speedup versus thread count for com-orkut and
com-friendster, reaching 18.9x and 23.1x at 32 threads respectively.}
\caption{Strong scaling (streaming phase, fixed stream order):
$18.9\times$ on \emph{com-orkut} and $23.1\times$ on the
$1.8$-billion-edge \emph{com-friendster} at $32$ threads.}
\label{fig:scaling}
\vspace{-0.5\baselineskip}
\end{figure}
The measured streaming phase is consistent with the prefix cost. With
$8$ threads at a $5\%$ prefix, the median streaming phase (prefix
ingestion plus all checkpoint counts, excluding input parse and
shuffle) takes $0.16$ seconds on \emph{sx-stackoverflow}, $0.53$
seconds on \emph{com-orkut}, and $10.8$ seconds on the
$1.8$-billion-edge \emph{com-friendster} --- an effective rate of
$1.7\times10^{8}$ stream edges per second on the largest graph. The
(corrected) MV20 implementation, which must read the whole stream and
build its sketch, takes $200$ seconds per trial on \emph{com-orkut} at
the same thread count. Early stopping helps twice: the
algorithm skips the remaining $m-S$ edges, and after materializing the
prefix it no longer has to maintain an online triangle estimator through
sequential edge updates. The expensive work is instead triangle
enumeration on a static prefix, so GBBS~\cite{dhulipala2018theoretically}
parallelizes the edge-local
neighbor-intersection counts and reduces the partial counts. The
stream-order scan and the search over prefix lengths remain sequential
across each graph; the dominant per-graph triangle count is what scales.
\cref{fig:scaling} shows this parallel side:
with the stream order fixed (so the work is identical at every thread
count), the GBBS implementation reaches an $18.9\times$ speedup at $32$
threads on \emph{com-orkut} and $23.1\times$ at $32$ threads on
\emph{com-friendster} ($72\%$ parallel efficiency); the larger graph
scales better because its prefix exposes more parallel
triangle-enumeration work per sequential byte.


\section{Conclusion}\label{sec:conclusion}

The threshold algorithm reframes streaming triangle counting as a
stopping problem rather than a budgeting problem: read until $Q$
triangles appear in the prefix, then return $\widehat T=Q(m/S)^3$.
Under $\eta\le T^{2/3}$, this gives a $(1\pm\varepsilon)$
approximation using $\widetilde O(m/T^{1/3})$ space, without knowing
$T$ or a memory budget in advance. Stopping is an algorithmic resource. While reservoir samplers remain strong when accuracy is measured after a full pass, they are inefficient when stream access is a scarce resource. The threshold algorithm self-selects the prefix, returns a whole-graph estimate after
seeing only $0.46\%$--\revfix{$6.6\%$} of the stream, and on a
$1.8\times10^9$-edge graph reads $0.46\%$ of the edges with $3.8\%$
error. A small random prefix can be enough, especially when you
know when to stop.

\clearpage

\balance
\makeatletter
\global\@ACM@balancefalse
\makeatother

\bibliographystyle{ACM-Reference-Format}
\bibliography{triangles}

\end{document}